\documentclass[runningheads,a4paper,10pt]{llncs}

\usepackage{amssymb,bm}
\usepackage{graphicx}
\usepackage{color}
\usepackage{algorithmicx}
\usepackage{algpseudocode}
\usepackage{pgfplots}
\usepgfplotslibrary{fillbetween}
\pgfplotsset{width=6cm,compat=newest}

\usepackage{bm}
\usepackage{url}
\usepackage{empheq}

\newcommand{\keywords}[1]{\par\addvspace\baselineskip
\noindent\keywordname\enspace\ignorespaces#1}

\usepackage{theorem,latexsym,graphicx,amssymb}
\usepackage{amsmath,enumerate}

\newtheorem{ex}{Example}

\begin{document}
\titlerunning{An Efficient Primal-Dual Algorithm for Fair Combinatorial Optimization}
\mainmatter
\title{%
An Efficient Primal-Dual Algorithm for Fair Combinatorial Optimization Problems}
\author{%
Viet Hung Nguyen\inst{1} \and Paul Weng\inst{2}
}
\institute{Sorbonne Universit\'es, UPMC Univ Paris 06, UMR 7606, LIP6, Paris, France \email{Hung.Nguyen@lip6.fr}
\and
SYSU-CMU Joint Institute of Engineering, Guangzhou, China\\
SYSU-CMU Joint Research Institute, Shunde, China\\
School of Electronics and Information Technology, SYSU\\
\email{paweng@mail.sysu.edu.cn}
}

\maketitle
\begin{abstract}

We consider a general class of combinatorial optimization problems including among others allocation, multiple knapsack, matching or travelling salesman problems. The standard version of those problems is the maximum weight optimization problem where a sum of values is optimized. However, the sum is not a good aggregation function when the fairness of the distribution of those values (corresponding for example to different agents' utilities or criteria) is important. In this paper, using the generalized Gini index (GGI), a well-known inequality measure, instead of the sum to model fairness, we formulate a new general problem, that we call fair combinatorial optimization. Although GGI is a non-linear aggregating function, a $0,1$-linear program (IP) can be formulated for finding a GGI-optimal solution by exploiting a linearization of GGI proposed by Ogryczak and Sliwinski \cite{OgryczakSliwinski03}. However, the time spent by commercial solvers (e.g., CPLEX, Gurobi...) for solving (IP) increases very quickly with instances' size and can reach hours even for relatively small-sized ones.  As a faster alternative, we propose a heuristic for solving (IP) based on a primal-dual approach using Lagrangian decomposition. 
We demonstrate the efficiency of our method by evaluating it against the exact solution of (IP) by CPLEX on several fair optimization problems related to matching. 
The numerical results show that our method outputs in a very short time efficient solutions giving lower bounds that CPLEX may take several orders of magnitude longer to obtain. Moreover, for instances for which we know the optimal value, these solutions are quasi-optimal with optimality gap less than 0.3\%.

%
\vspace{-3mm}
\end{abstract}

\vspace{-10mm}
\section{Introduction}
\label{sec:intro}
 
%
%
%
%
%
%
%
\vspace{-3mm}
The solution of a weighted combinatorial optimization problem can be seen as the selection of $n$ values in a combinatorial set $\mathcal X \subset \mathbb R^n$.
The maximum weight version of such a problem consists in maximizing the sum of these $n$ values (e.g., $\sum_{i=1}^n u_i$).
For instance, in a matching problem on a graph, the sum of weights that is optimized corresponds to the sum of weights of the edges selected in a matching.
In practice, the vector of weights $(u_1, u_2, \ldots, u_n)$ could receive different interpretations depending on the actual problem.
In a multi-agent setting, each value $u_i$ represents the utility of an agent $i$, as in a bi-partite matching  problem where $n$ objects have to be assigned to $n$ agents.
In a multi-criteria context, those $n$ values can be viewed as different dimensions to optimize.
For example, in the travelling salesman problem (TSP) with $n$ cities, a feasible solution (i.e., Hamiltonian cycle) is valued by an $n$-dimensional vector where each component represents the sum of the distances to reach and leave a city.

\vspace{-.5mm}
In both interpretations, it is desirable that the vector of values $(u_1, u_2, \ldots, u_n)$ be both Pareto-optimal (i.e., not improvable on all components at the same time) and balanced (or fair).
We call optimization with such concerns {\it fair optimization} by adopting the terminology from multi-agent systems.
In this paper, we focus on the fair optimization version of a class of combinatorial problems (including allocation, general matching, TSP...). 
Note that optimizing the sum of the values (i.e., maximum weight problem) yields a Pareto-optimal solution, but does not provide any guarantee on how balanced the vector solution would be.

\vspace{-.5mm}
Various approaches have been proposed in the literature to provide such a guarantee with different models for fairness or "balancedness" (see Section~\ref{sec:related} for an overview).
In this paper, our approach is based on an inequality measure called {\em Generalized Gini Index} (GGI) \cite{Weymark81}, which is well-known and well-studied in economics and can be used to control for both Pareto-efficiency and fairness. 
Indeed, fairness has naturally been investigated in economics \cite{Moulin04}.
In this literature, two important requirements have been identified as essential for fairness: equal treatment of equals and efficiency.
The first notion implies that two agents with the same characteristics (notably the same preferences) have to be treated the same way, while the second entails that a fair solution should be Pareto-optimal.
GGI satisfies both requirements, as it is symmetric in its arguments and increasing with Pareto dominance.
The notion of fairness that GGI encodes is based on the Pigou-Dalton transfer principle, which states that a small transfer of resource from a richer agent to a poorer one yields a fairer distribution.

\vspace{-.5mm}
To the best of our knowledge, fair optimization in such general combinatorial problems has not been considered so far, although the GGI criterion has been investigated before in some specific problems (allocation \cite{LescaPerny11}, capital budgeting \cite{KOWejor04}, Markov decision process \cite{OgryczakPernyWeng11ADT,OgryczakPernyWeng13}...).
The difficulty of this combinatorial optimization problem lies in the fact that the objective function is non-linear.
The contribution of this paper is fourfold:
(1) we introduce a new general combinatorial problem (e.g., fair matching in general graph or fair TSP have not been studied so far);
(2) we provide an optimality condition and an approximation ratio; 
(3) we propose a fast general heuristic method based on a primal-dual approach and on Lagrangian decomposition;
(4) we evaluate this method on several problems related to matching to understand its efficiency.
Although our general combinatorial formulation covers problems whose maximum weight version is NP-hard, we leave for a follow-up work the integration of our fast heuristic with approximation algorithms to solve those NP-hard problems.

\vspace{-.5mm}
The paper is organized as follows.
Section~\ref{sec:related} gives an overview of related work.
Section~\ref{sec:model} provides a formal definition of our problem, which can be solved by a $0,1$-linear program.
As a faster alternative, we present a heuristic primal-dual solving method based on Lagrangian decomposition in Section \ref{sec:method}  and evaluate it experimentally in Section~\ref{sec:expe}.
Finally, we conclude in Section~\ref{sec:conclu}.


\vspace{-2mm}
\section{Related Work}
\label{sec:related}

\vspace{-2mm}
Fair optimization is an active and quite recent research area \cite{OgryczakLussPioroNaceTomaszewski14,Luss12} in 
multiobjective optimization.
Fairness can be modeled in different ways.
One simple approach is based on maxmin, so called Egalitarian approach, where one aims at maximizing the worse-off component (i.e., objective, agent...).
Due to the drowning effects of the min operator, vectors with the same minimum cannot be discriminated.
A better approach \cite{Rawls71} is based on the lexicographic maxmin, which consists in considering the minimum first when comparing two vectors, then in case of a tie, focusing on the second smallest values and so on.
However, due to the noncompensatory nature of the min operator, vector $(1, 1, \ldots, 1)$ would be preferred to $(0, 100, \ldots, 100)$, which may be debatable.
To take into account this observation, one can resort to use a strictly increasing and strictly Schur-concave (see Section~\ref{sec:owa} for definition) aggregation function $f$ (see \cite{OgryczakLussPioroNaceTomaszewski14} for examples) that evaluates each vector such that higher values are preferred.

In this paper, we focus on the Generalized Gini Index (GGI) proposed in the economics literature \cite{Weymark81}, because it satisfies natural properties for encoding fairness.
GGI is a particular case of a more general family of operators known as Ordered Weighted Averaging (OWA) \cite{Yager88}.
Much work in fair optimization has applied the OWA operator and GGI in multiobjective (continuous and combinatorial) optimization problems.
To cite a few, it was used in network dimensioning problems \cite{OgryczakSliwinskiWierzbicki03}, capital budgeting \cite{KOWejor04}, allocation problems \cite{LescaPerny11}, 
flow optimization in wireless mesh networks \cite{HurkalaSliwinski12} and multiobjective sequential decision-making under uncertainty \cite{OgryczakPernyWeng11ADT,OgryczakPernyWeng13}.
One common solving technique is based on a linearization trick of the nonlinear objective function based on GGI \cite{OgryczakSliwinski03}.
Recently, \cite{GilbertSpanjaard17} considered a similar setting to ours, but tries to solve its continuous relaxation.

In multicriteria decision-making, fair optimization is related to compromise optimization, which generally consists in minimizing a distance to an ideal point \cite{Steuer86}. 
More generally, the ideal point can be replaced by any reference point that a decision maker chooses, as in the reference point method \cite{Wierzbicki82}.
In this context, a judiciously chosen reference point can help generate a solution with a balanced profile on all criteria.
One main approach is based on minimizing the augmented weighted Tchebycheff distance.
This method has been applied in many multicriteria problems, for instance, in process planning \cite{RoderaBagajewiczTrafalis02}, in sequential decision-making under uncertainty \cite{PernyWeng10ECAI}, in discrete bicriteria optimization problems \cite{DachertGorskiKlamroth12}, in multiobjective multidimensional knapsack problems \cite{LustTeghem12}.


Note that our combinatorial optimization problem should not be confused with the multicriteria version of those problems where each scalar weight becomes vectorial and the value of a solution is obtained by aggregating the selected weight vectors with a componentwise sum.
For instance, Anand \cite{Anand06} investigated a multicriteria version of the matching problem and proved that the egalitarian approach for vector-valued matching leads to NP-hard problems.
In our problem, the weights are scalar and the value of a solution is {\em not} obtained by summing its scalar weights, but by aggregating them with GGI. 


\vspace{-3mm}
\section{Model}
\label{sec:model}

\vspace{-2mm}
In this section, we formally describe the general class of combinatorial problems considered in this paper and provide some concrete illustrative examples in this class.
Then we recall the generalized Gini index as a measure of fairness and define the fair combinatorial optimization problems tackled in this paper.
We start with some notations.
For any integer $n$, $[n]$ denotes the set $\{1, 2, \ldots n\}$.
For any vector $\bm x$, its component is denoted $x_i$ or $x_{ij}$ depending on its dimension.

\vspace{-2mm}
\subsection{General Model}

We consider a combinatorial optimization problem (e.g., allocation, multiple knapsack, matching, travelling salesman problem...), whose feasible solutions $\mathcal X \subseteq \{0, 1\}^{n \times m}$ can be expressed as follows:

\vspace{-6mm}
\begin{align*}
& \bm A \bm z \le \bm b  \\
& \bm z \in \{0, 1\}^{n \times m} 
\end{align*}

\vspace{-3mm}
\noindent where $\bm A \in \mathbb Z^{p \times (nm)}$, $\bm b \in \mathbb Z^{p}$, $n$, $m$ and $p$ are three positive integers,  and $\bm z$ is viewed as a one-dimensional vector $(z_{11}, \ldots, z_{1m}$, $z_{21}, \ldots, z_{2m}$, $\ldots$, $z_{n1}, \ldots, z_{nm})^\intercal$.

Let $u_{ij} \in \mathbb N$ be the utility of setting $z_{ij}$ to $1$.
The {\em maximum weight problem} defined on combinatorial set $\mathcal X$ can be written as a $0,1$-linear program ($0,1$-LP):

\vspace{-1mm}
\begin{subequations}
\small
\begin{empheq}[]{alignat=2}
{\max}. ~ & \sum_{i \in [n]} \sum_{j \in [m]} u_{ij} z_{ij} \nonumber\\
\mbox{s.t.~~} & \bm z \in \mathcal X \nonumber
\end{empheq}
\end{subequations}

\vspace{-0mm}
\noindent Because this general problem includes the travelling salesman problem (TSP), it is NP-hard in general.
As mentioned before, this objective function provides no control on the fairness of the obtained solution.
Although possibly insufficient, one simple approach to fairness consisting in focusing on the worse-off component is the {\em maxmin problem} defined on set $\mathcal X$, which can also be written as a $0,1$-LP:

\vspace{-1mm}
\begin{subequations}
\small
\begin{empheq}[]{alignat=2}
{\max}. ~ & v \nonumber\\
\mbox{s.t.~~} & v \le \sum_{j \in [m]} u_{ij} z_{ij} &\quad& \forall i \in [n] \nonumber\\
& \bm z \in \mathcal X \nonumber
\end{empheq}
\end{subequations}

\vspace{-0mm}
\noindent Even for some polynomial problems like allocation, this version is NP-hard in general \cite{BezakovaDani05}.
To avoid any confusion, in this paper, allocation refers to matching on a bi-partite graph and matching generally implies a complete graph.

For illustration, we now present several instantiations of our general model on allocation and matching problems, some of which will be used for the experimental evaluation of our proposed methods in Section~\ref{sec:expe}.

\vspace{-3mm}
\begin{ex}[Allocation]
Let $G=(V_1 \cup V_2, E, u)$ be a valued bipartite graph where $V_1$ and $V_2$ are respectively an $n$-vertex set and an $m$-vertex set with $V_1 \cap V_2 = \emptyset$, $E \subseteq \big\{\{x, y\} \,|\, (x, y) \in V_1\times V_2\big\}$ is a set of non-directed edges and $u : E \to \mathbb R$ defines the nonnegative utility (i.e., value to be maximized) of an edge.
As there is no risk of confusion, we identify $V_1$ to the set $[n]$ and $V_2$ to the set $[m]$.
An {\em allocation} of $G$ is a subset of $E$ such that each vertex $i$ in $V_1$ is connected to $\alpha_{i}$ to $\beta_{i}$ vertices in $V_2$ and each vertex in $V_2$ is connected to $\alpha'_{j}$ to $\beta'_{j}$ vertices in $V_1$ where
$(\bm\alpha, \bm\beta) \in \mathbb N^{n \times n}$ and $(\bm\alpha', \bm\beta') \in \mathbb N^{m \times m}$.

The {\em assignment} problem where $n$ tasks need to be assigned to $n$ agents is a special case where $n=m$ and $\alpha_{i}=\beta_{i}=\alpha'_{j}=\beta'_{j}=1$ for $i \in [n]$ and $j \in [n]$.
The {\em conference paper assignment} problem where $m$ papers needs to be reviewed by $n$ reviewers such that each paper is reviewed by $3$ reviewers and each reviewer receives at most $6$ papers can be represented with $\alpha_i=0$, $\beta_i=6$, $\alpha'_j=3$ and $\beta'_j=3$ for $i \in [n]$ and $j \in [m]$.
The {\em Santa Claus} problem \cite{BansalSviridenko06} where $m$ toys needs to be assigned to $n$ children with $n \le m$ is also a particular case with $\alpha_i=0$, $\beta_i=m$, $\alpha'_j=\beta'_j=1$ for $i \in [n]$ and $j \in [m]$.

The maximum weight problem can be solved with the following $0,1$-LP:

\vspace{-2mm}
\begin{subequations}
\small
\begin{empheq}[]{alignat=2}
{\max}. ~ & \sum_{i \in [n]} \sum_{j \in [m]} u_{ij} z_{ij} \nonumber \\
\mbox{s.t.~~} & \alpha_{i} \le \sum_{j \in [m]} z_{ij} \le \beta_{i} &~& \forall i \in [n] \label{eq:ac1}\\
& \alpha'_{j} \le \sum_{i \in [n]} z_{ij} \le \beta'_{j} &~& \forall j \in [m] \label{eq:ac2}\\
&\bm z \in \{0, 1\}^{n \times m} \nonumber 
\end{empheq}
\end{subequations}

\vspace{-0mm}
\noindent Interestingly, its solution can be efficiently obtained by solving its continuous relaxation because the matrix defining its constraints \eqref{eq:ac1}--\eqref{eq:ac2} is totally unimodular \cite{Schrijver98}.
However, the maxmin version is NP-complete \cite{BezakovaDani05}.
\end{ex}

\vspace{-5mm}
\begin{ex}[Matching]
Let $G=(V, E, u)$ be a valued graph where $V$ is a $2n$-vertex set (with $n \in \mathbb N \backslash \{0\}$), $E \subseteq \big\{\{x, y\} \,|\, (x, y) \in V^2, x \neq y\big\}$ is a set of non-directed edges and $u : E \to \mathbb R$ defines the nonnegative utility 
of an edge.
A {\em matching} $M$ of $G$ is a subset of $E$ such that no pair of edges of $M$ are adjacent, i.e., they do not share a common vertex: $\forall (e, e') \in E^2, e \neq e' \Rightarrow e \cap e' = \emptyset$.
A {\em perfect matching} $M$ is a matching where every vertex of $G$ is incident to an edge of $M$.
Thus, a perfect matching contains $n$ edges.
Without loss of generality, we identify $V$ to the set $[2n]$ and denote $\forall e=\{i, j\} \in E, u_{ij} = u(e)$ when convenient.


The standard maximum weight perfect matching problem aims at finding a perfect matching for which the sum of the utilities of its edges is maximum.
Let $\delta(i)=\{\{i, j\}\in E \,|\, j \in V\backslash \{i\}\}$ be the set of edges that are incident on vertex $i$.
It is known \cite{LovaszPlummer86} that this problem can be formalized as a $0,1$-LP (where $z_{ij}$'s for $i>j$ are unnecessary and can be set to $0$): 

\vspace{-4mm}
\begin{subequations}
\small
\begin{empheq}[left={(\mathcal P_M) }\empheqlbrace]{align}
     {\max}. \quad & \displaystyle\sum_{i \in [2n]} \sum_{j \in [2n], j>i} u_{ij} z_{ij} &  \label{eq:pmo}\\
     \mbox{s.t.} \quad & \displaystyle\sum_{\{i, j\} \in \delta(k), i<j} z_{ij} = 1 & \forall k \in [2n] \label{eq:pmc1} \\
            & z_{ij} \in \{0, 1\} & \forall i \in [2n], j = i+1, \ldots, 2n \label{eq:pmc2}
\end{empheq}
\end{subequations}
where \eqref{eq:pmc1} states that in a matching only one edge is incident on any vertex.

This problem can be solved as an LP by considering the continuous relaxation of $\mathcal P_M$ and adding the well-known blossom constraints \eqref{eq:mc2} in order to remove the fractional solutions introduced by the relaxation:

\vspace{-1mm}
\begin{subequations}
\small
\begin{empheq}[left={(\mathcal {RP}_M) }\empheqlbrace]{alignat=2}
    {\max}. \quad  & \displaystyle\sum_{i \in [2n]}\sum_{j \in [2n], j>i} u_{ij} z_{ij} && \nonumber \\
    \mbox{s.t.} \quad & \displaystyle\sum_{\{i,j\} \in \delta(k), i<j} z_{ij} = 1 &\quad& \forall k \in [2n] \label{eq:mc1} \\
            & z(\delta(S)) \ge 1 && \forall S \subset V, |S| \mbox{ odd}, |S| \ge3 \label{eq:mc2}\\
            & 0 \le z_{ij} \le 1 && \forall i \in [2n], j = i+1, \ldots, 2n \label{eq:mc3}
\end{empheq}
\end{subequations}

\vspace{0mm}
\noindent where $z(\delta(S)) = \sum_{\{i,j\} \in \delta(S), i<j} z_{ij}$ and $\delta(S)=\{\{i, j\}\in E \,|\, i \in S \mbox{~and } j \in V\backslash S\}$.
Constraints \eqref{eq:mc1}--\eqref{eq:mc3} define the so-called {\em perfect matching polytope}.
In practice, this problem can be efficiently solved with the Blossom algorithm proposed by Edmonds \cite{Edmonds65}. 
To the best of our knowledge, the maxmin version of the matching problem (on complete graph) has not been investigated so far.
\end{ex}

In this paper we focus on a variant of those combinatorial problems:
search for a solution $\bm z$ whose distribution of values $\big( \sum_{j \in [m]} u_{ij} z_{ij} \big)_{i \in [n]}$ is fair to its components (e.g., different agents' utilities or criteria).
To model fairness we use a special case of the ordered weighted averaging operator that we recall next.

\vspace{-3mm}
\subsection{Ordered Weighted Average and Generalized Gini Index}\label{sec:owa}

\vspace{-1mm}
The {\em Ordered Weighted Average} (OWA) \cite{Yager88} of $\bm v \in \mathbb R^n$ is defined by:
\begin{align*}
OWA_{\bm w}(v) = \sum_{k \in [n]} w_k v^\uparrow_k
\end{align*}

\vspace{-2mm}
\noindent where $\bm{w} = (w_1, \ldots, w_n) \in [0, 1]^n$ is the OWA weight vector and $\bm v^\uparrow = (v^\uparrow_1, \ldots, v^\uparrow_n)$ is the vector obtained from $\bm{v}$ by rearranging its components in an increasing order.
OWA defines a very general family of operators, e.g., the sum (for $w_k=1$, $\forall k \in [n]$), the average, the minimum (for $w_1=1$ and $w_k=0$, $\forall k > 1$), 
the maximum (for $w_n=1$ and $w_k=0$, $\forall k < n$), the leximin when differences between OWA weights tends to infinity or the augmented weighted Tchebycheff distance \cite{OgryczakPernyWeng13}.

Let the {\em Lorenz components} \cite{Arnold87} of $\bm v$ be denoted by $(L_1(\bm v)$, $\ldots$, $L_n(\bm v))$ and be defined by $\forall k \in [n]$, $L_k(\bm v) = \sum_{i \in [k]} v^\uparrow_i$.
Interestingly, OWA can be rewritten as: 

\vspace{-4mm}
\begin{align}\label{eq:owaLorenz}
OWA_{\bm w}(\bm v) = \sum_{k \in [n]} w'_k L_k(\bm v)
\end{align}

\vspace{-3mm}
\noindent where $\forall k \in [n], w'_k = w_k - w_{k+1}$ and $w_{n+1} = 0$.
With this rewriting, one can see that OWA is simply a weighted sum in the 
space of Lorenz components.

The notion of fairness that we use in this paper is based on the {\em Pigou-Dalton principle} \cite{Moulin88}. 
It states that, all other things being equal, we prefer more ``balanced'' vectors, which implies that any transfer (called {\em Pigou-Dalton transfer}) from a richer component to a poorer one without reversing their relative positions yields a preferred vector.
Formally, for any $\bm v \in \mathbb R^n$ where $v_i < v_j$ and for any $\epsilon \in (0, v_j -v_i)$ we prefer $\bm v+\epsilon \mathbf 1_i - \epsilon \mathbf 1_j$ to $\bm v$ where $\mathbf 1_i$ (resp. $\mathbf 1_j$) is the canonical vector, null everywhere except in component $i$ (resp. $j$) where it is equal to $1$.

When the OWA weights are strictly decreasing and positive \cite{Weymark81}, OWA is called the {\em Generalized Gini Index} (GGI) \cite{Weymark81} and denoted $G_{\bm w}$.
It encodes both:\\[1ex]
\noindent {\bf efficiency:} $G_{\bm w}$ is increasing with respect to Pareto-dominance (i.e., if $\bm v\in \mathbb R^n$ Pareto-dominates\footnote{Vector $\bm v$ Pareto-dominates vector $\bm v'$ if $\forall i \in [n], v_i \ge v'_i$ and $\exists j \in [n], v_j > v'_j$.} $\bm v' \in \mathbb R^n$, then $G_{\bm w}(\bm v) > G_{\bm w}(\bm  v')$); and \\[1ex]
{\bf fairness:} $G_{\bm w}$ is strictly Schur-concave, i.e., it is strictly increasing with Pigou-Dalton transfers ($\forall \bm v\in\mathbb R^n$, $v_i < v_j$, $\forall \epsilon \in (0, v_j -v_i), G_{\bm w}(\bm v+\epsilon \mathbf 1_i - \epsilon \mathbf 1_j) > G_{\bm w}(\bm v)$).\\

\vspace{-3mm}
The classic Gini index, which is a special case of GGI with $w_i = (2(n -i)+1)/n^2$ for all $i \in [n]$, enjoys a nice graphical interpretation (see Figure~\ref{fig:lorenzcurve}).
For a given distribution $\bm v \in \mathbb R^n_+$, let $\bar v$ denote the average of the components of $\bm v$, i.e., $\bar v = \frac{1}{n} \sum_{i=1}^n v_i$.
Distribution $\bm v$ can be represented by the curve going through the points $(0, 0)$ and $(\frac{k}{n}, L_k(\bm v))$ for $k \in [n]$.
The most equitable distribution with the same total sum as that of $\bm v$ (i.e., $n \bar v$) can be represented by the straight line going through the points $(0, 0)$ and $(\frac{k}{n}, k\bar v)$ for $k \in [n]$.
The value $1 - G_{\bm w}(\bm v)/\bar v$ is equal to twice the area between the two curves.
\begin{figure}[tb!]
\centering
\scalebox{.8}{
\begin{tikzpicture}
\begin{axis}[
    xmin=0, xmax=100,
    ymin=0, ymax=100,
    grid,
    ylabel = cumulative value,
    xlabel = cumulative population (in \%),
    legend style={legend pos=north west},
    ]
\addplot+[name path=A] plot 
    coordinates { (0,0) (25,10) (50,20) (75,50) (100,100)};
\addplot+[name path=B] plot 
    coordinates { (0,0) (25, 25) (50, 50) (75, 75) (100,100)};
\addplot+[gray, fill opacity=0.4] fill between[of=A and B];
\end{axis}
\end{tikzpicture}}
\caption{Lorenz curves}
\label{fig:lorenzcurve}
\end{figure}
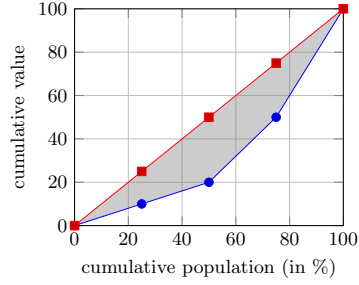

Interestingly, the Lorenz components of a vector can be computed by LP \cite{OgryczakSliwinski03}.
Indeed, the $k$-th Lorenz component $L_k(\bm v)$ of a vector $\bm v$ can be found as the solution of a knapsack problem, which is obtained by solving the following LP:

\vspace{-2mm}
\begin{subequations}
\small
\begin{empheq}[left={(\mathcal{LP}_k) }\empheqlbrace]{align}
{\min}. \quad &\quad \sum_{i \in [n]} a_{ik} x_i & \nonumber\\
\mbox{s.t.} \quad &\quad \sum_{i \in [n]} a_{ik} = k & \nonumber\\
&\quad 0 \le a_{ik} \le 1 & \forall i \in [n] \nonumber
\end{empheq}
\end{subequations}
Equivalently, this can be solved by its dual:
\begin{subequations}
\small
\begin{empheq}[left={(\mathcal{DL}_k) }\empheqlbrace]{align}
{\max}. \quad &\quad k r_k - \sum_{i \in [n]} d_{ik} & \nonumber\\
\mbox{s.t.} \quad &\quad r_k - d_{ik} \le v_i & \forall i \in [n] \nonumber\\
&\quad d_{ik} \ge 0 & \forall i \in [n] \nonumber
\end{empheq}
\end{subequations}
The dual formulation is particularly useful.
Contrary to the primal, it can be integrated in an LP where the $v_i$'s are also variables \cite{OgryczakSliwinski03}.
We will use this technique to formulate a $0,1$-LP to solve our general combinatorial optimization problem.

\subsection{Fair Combinatorial Optimization}

The problem tackled in this paper is defined 
by using GGI as objective function:
\begin{align*}
{\max}. ~ G_{\bm w}\big( \big( \sum_{j \in [m]} u_{ij} z_{ij} \big)_{i \in [n]} \big) & \quad
\mbox{s.t.~~} \left\{
\begin{array}{l}
\bm A \bm z \le \bm b \\
\bm z \in \{0, 1\}^{n \times m}
\end{array}\right.
\end{align*}
Following Ogryczak and Sliwinski \cite{OgryczakSliwinski03}, we can combine the rewriting of OWA based on Lorenz components \eqref{eq:owaLorenz} and LPs $(\mathcal{DL}_k)$ for $k \in [2n]$ to transform the previous non-linear optimization program into a $0,1$-LP:
\begin{subequations}
\small
\label{eq:primal}
\begin{empheq}[]{align}
{\max}. ~ & \sum_{k \in [n]} w'_k ( k r_k - \sum_{i \in [n]} d_{ik} )\label{eq:primal1} \\
\mbox{s.t.~~} & \bm A \bm z \le \bm b && \\
&\bm z \in \{0, 1\}^{n \times m} \\
& r_k - d_{ik} \le \sum_{j \in [m]} u_{ij} z_{ij} && \forall i \in [n], \forall k \in [n] \label{eq:primalc}\\
& d_{ik} \ge 0 && \forall i \in [n], \forall k \in [n] \label{eq:primal2}
\end{empheq}
\end{subequations}
Due to the introduction of new constraints \eqref{eq:primalc}--\eqref{eq:primal2} from LPs $(\mathcal{DL}_k)$, the relaxation of this $0,1$-LP may yield fractional solutions.
The naive approach to solve it would be to give it to a $0,1$-LP solver (e.g., Cplex, Gurobi...).
Our goal in this paper is to propose an adapted solving method for it, which would be much faster than the naive approach by exploiting the structure of this problem.

\section{Alternating Optimization Algorithm} \label{sec:method}

Before presenting our approach, which is a heuristic method based on a primal-dual technique using a Lagrangian decomposition, we first make an interesting and useful observation.
The dual of the continuous relaxation of the previous $0,1$-LP \eqref{eq:primal} is given by:
\begin{subequations}
\small
\label{eq:dual}
\begin{empheq}[]{align}
{\min}. ~ &  \bm b^\intercal \bm v + \sum_{i \in [n]} \sum_{j \in [m]} t_{ij} \label{eq:dual1}\\
\mbox{s.t.~~} & (\bm v^\intercal \bm A)_{ij} +t_{ij} - \sum_{k \in [n]} u_{ij} y_{ik} \ge 0 && \forall i \in [n], \forall j \in [m]\\
& \sum_{i=1}^n y_{ik} = k w'_k && \forall k \in [n]\\
& 0 \le y_{ik} \le w'_k && \forall i \in [n], \forall k \in [n]\\
& v_{j} \ge 0 && \forall j \in [p]\\
& t_{ij} \ge 0 && \forall i \in [n], \forall j \in [m] \label{eq:dual2}
\end{empheq}
\end{subequations}
Interestingly, with fixed $y_{ik}$'s, the dual of the previous program can be written in the following form, which is simply the continuous relaxation of the original program with modified weights:
\begin{subequations}
\small
\label{eq:primalrelaxation}
\begin{empheq}[]{align}
{\max}. ~ & \sum_{i \in [n]} \big(\sum_{k \in [n]} y_{ik} \big) \sum_{j \in [m]} u_{ij}z_{ij} \label{eq:primalrelaxation1} \\
\mbox{s.t.~~} & \bm A \bm z \le \bm b && \\
&\bm z \in [0, 1]^{n \times m} \label{eq:primalrelaxation2}
\end{empheq}
\end{subequations}
Therefore, solving this program with discrete $\bm z$ yields a feasible solution of the original problem.
We denote ($P_{\bm y}$) the $0,1$-LP \eqref{eq:primalrelaxation} defined with $\bm y = (y_{ik})_{i \in [n], k \in [n]}$.

\subsection{Optimality Condition and Approximation Ratio}

Next we express an optimality condition so that an integer solution $\bm z^*$ computed from a dual feasible solution $\bm y^*$ of \eqref{eq:dual} is optimal for program \eqref{eq:primal}.
First, note that any extreme solution $(\bm v, \bm t, \bm y)$ of program \eqref{eq:dual} is such that either $y_{ik}=0$ or $y_{ik}=w'_k$ for all $i \in [n]$ and $k \in [n]$.
\begin{theorem}\label{th:optimal}
Let $(\bm v, \bm t, \bm y^*)$ be an extreme solution of \eqref{eq:dual} and let $\bm z^*$ be the optimal solution of program {\em ($P_{\bm y^*}$)}. 
Let $T^*_i=\sum_{j \in [m]} u_{ij}z^*_{ij}$ for all $i \in [n]$ and assume without loss of generality that $T^*_{1} \geq T^*_{2} \geq \ldots \geq T^*_n$. \\
If for all $k \in [n]$, $y^*_{ik}=w'_k$ for all $i\geq n+1-k$ and $y^*_{ik}=0$ for all $i \in [n-k]$ then $\bm z^*$ is an optimal solution of program \eqref{eq:primal}.
\end{theorem}
\begin{proof}
Let $(\bm v^*, \bm t^*)$ be the dual optimal solution associated with $\bm z^*$ when solving ($P_{\bm y^*}$). 
Composing them with $\bm y^*$, we obtain a feasible solution $(\bm v^*, \bm t^*, \bm y^*)$ of \eqref{eq:dual}. 
By duality theory of linear programming, the objective value of this solution is equal to $\sum_{i \in [n]}(\sum_{j \in [i]} w'_{n+1-j})T^*_i$. 
Let us now build a feasible solution $(\bm r^*, \bm d^*, \bm z^*)$  of \eqref{eq:primal} based on $\bm z^*$ as follows. 
For all $k \in [n]$, \\[.5ex]
$\bullet$ $r^*_k=T^*_{n+1-k}$ and\\[.5ex]
$\bullet$ $d^*_{ik}=\left\{ \begin{array}{cc}r^*_k - T^*_{i} & \mbox{if $i \geq n+1- k$}\\
0 & \mbox{otherwise} \end{array} \right.$ for all $i \in [n]$ .\\[.5ex]
We now show that $(\bm r^*, \bm d^*)$ satisfy constraints \eqref{eq:primalc}. 
For any $i \in [n]$ and $k \in [n]$, if $i \leq n+1-k$ then as $r^*_k=T^*_{n+1-k}\leq T^*_{i}$ and $d^*_{ik}=0$, we have 
\begin{align*}
r^*_k -d^*_{ik} \leq T^*_i= \sum_{j \in [m]} u_{ij}z^*_{ij} \enspace .
\end{align*}
If $i\geq n+1-k$ then as $d^*_{ik}=r^*_k - T^*_{i} $, $r^*_k -d^*_{ik}=T^*_i=\sum_{j \in [m]} u_{ij}z^*_{ij}$.
Hence $(\bm r^*, \bm d^*, \bm z^*)$ is a feasible solution of \eqref{eq:primal}.
For any $k \in [n]$, $k r^*_k - \sum_{i \in [n]} d^*_{ik}=kr^*_k - \sum_{i=n+1-k}^nd^*_{ik}=
kr^*_k - (kr^*_k - \sum_{i=n+1-k}^nT^*_i)=\sum_{i=n+1-k}^nT^*_i)$. 
Then it is easy to see that the objective value of this solution, which is 
$\sum_{k \in [n]} w'_k ( k r^*_k - \sum_{i \in [n]} d^*_{ik})$ is equal to $\sum_{k \in [n]} w'_{k}\sum_{i=n-k+1}^nT^*_i$.
This sum is just a rewriting of $\sum_{i \in [n]}(\sum_{j \in [i]} w'_{n+1-j})T^*_i$. 
Thus, by duality of linear programming, the solution $(\bm r^*, \bm d^*, \bm z^*)$ is optimal for program \eqref{eq:primal}.
\qed
\end{proof}
Theorem \ref{th:optimal} provides an optimality condition for any feasible solution $\bm z^*$, but does not indicate how to find "good" solutions. 
Yet, one may be interested in the quality of some special solutions, e.g., the optimal solution of the maximum weight version. 
The following theorem establishes an approximation ratio for the latter, which also applies to our method as discussed later.
\begin{theorem}\label{th:bound}
Let $\bar{\bm z}$ be an optimal solution of the maximum weight version.
Let $\bar{T}_i=\sum_{j \in [m]} u_{ij}\bar{z}_{ij}$ for all $i \in [n]$ and assume without loss of generality that $\bar{T}_{1} \geq \bar{T}_{2} \geq \ldots \geq \bar{T}_n$. Let $w'_{\max} = \max_{k \in [n]} w'_k$.
Then the GGI value of $\bar{\bm z}$ is at worst $\max(\frac{2w'_n}{(n+1)w'_{\max}}, \frac{n \bar{T}_n}{(\sum_{i \in [n]} \bar{T}_i))}) $ of the optimal objective value of program \eqref{eq:primal}.
\end{theorem}
\begin{proof}
Let vector $\bar{\bm y} \in \mathbb{R}^{n\times n}$ be defined as $\bar{y}_{ik}=\frac{k}{n}w'_k$ for $i, k \in [n]$, which is feasible for program \eqref{eq:dual}.  
The objective function of ($P_{\bar{\bm y}}$) satisfies: 
\begin{align}
\sum_{i \in [n]} \big(\sum_{k \in [n]} \bar{y}_{ik} \big) \sum_{j \in [m]} u_{ij}z_{ij} &= 
\sum_{i \in [n]} \sum_{j \in [m]} (\sum_{k \in [n]} \frac{k}{n}w'_k)u_{ij}z_{ij} \nonumber\\
&\le \sum_{i \in [n]} \sum_{j \in [m]} (\sum_{k \in [n]} \frac{k}{n}w'_{\max})u_{ij}z_{ij} \label{eq:wmax}
\end{align}
Program \eqref{eq:primalrelaxation} with objective~\eqref{eq:wmax} corresponds to the maximum weight version scaled by a constant. 
It is equal to $\sum_{k \in [n]} (kw'_{\max}/n) \times$ $\sum_{i \in [n]}(\bar{T}_i)$ for solution $\bar{\bm z}$, which is an upperbound of the objective value associated with $\bar{\bm y}$ of \eqref{eq:dual} and hence an upperbound for the optimal value of \eqref{eq:primal}.\\
Proceeding as for Theorem \ref{th:optimal}, we define a feasible solution of \eqref{eq:primal} based on $\bar{\bm z}$:\\[.5ex]
$\bullet$ $\bar{r}_k=\bar{T}_{n+1-k}$ for all $k \in [n]$, and \\[.5ex]
$\bullet$ $\bar{d}_{ik}=\left\{ \begin{array}{cc}\bar{r}_k - \bar{T}_{i} & \mbox{if $i \geq n+1- k$}\\
0 & \mbox{otherwise} \end{array} \right.$ for all $i \in [n]$, for all $k \in [n]$.\\[.5ex]
The objective value of this solution $\sum_{i \in [n]} (\sum_{j \in [i]} w'_{n+1-j})\bar{T}_i$ (see proof of Theorem \ref{th:optimal}) is to be compared with upperbound $\sum_{i \in [n]} (\sum_{k \in [n]} kw'_{\max}/n)\bar{T}_i$.\\
By comparing term by term w.r.t. $\bar{T}_i$ for $i \in [n]$, we can see that the worst case happens to the term associated with $\bar{T}_1$ with the ratio $w'_n/(\sum_{k \in [n]} k w'_{\max}/n)$. 
Therefore, we obtain the ratio $\frac{2w'_n}{(n+1)w'_{\max}}$. 
This ratio is consistent since when $n=1$, the optimal solution of the maximum weight version coincides with the optimum solution of \eqref{eq:primal}. \\
By comparing term by term with respect to $w'_k$ for $k \in [n]$, we can see that the worst case happens to the term associated with $w'_1$ with the ratio $n\bar{T}_n/(\sum_{i=1}^{n}\bar{T}_i)$,
which can be interpreted as the ratio of the smallest utility over the average utility in the optimal solution of the maximum weight version. 
This ratio is consistent since in the case of equal utilities in the optimal solution of the maximum weight version, the latter coincides with the optimum solution of \eqref{eq:primal}.
\qed
\end{proof}
\subsection{Iterative Algorithm}
The previous discussion motivates us to design an alternating optimization algorithm that starts with a feasible $\bm y$ for \eqref{eq:dual}, computes the associated $\bm z$ and uses the latter to iteratively improve $\bm y$. 
Formally, it can be sketched as follows:
\begin{algorithmic}[1]
\State $t \gets 0$
\State compute $\bm y^{(0)}$ \label{algl:initial}
\Repeat 
\State $t \gets t+1$
\State solve $0,1$-LP ($P_{\bm y^{t-1}}$) to obtain feasible solution $\bm z^{(t)}$
\State update $\bm y^{(t)}$ based on $\bm y^{(t-1)}$ and $\bm z^{(t)}$ \label{algl:update}
\Until {max iteration has been reached or change on  $y_{ik}^{(t)}$ is small}
\State return $\bm z^{(t)}$ with highest GGI
\end{algorithmic}
Interestingly, lines~\ref{algl:initial} and \ref{algl:update} can be performed in different ways.
For line~\ref{algl:initial}, an initial $\bm y^{(0)}$ can be obtained by solving the dual LP \eqref{eq:dual}.
Another approach is to solve the maximum weight version of our combinatorial problem and get the dual solution variables for $\bm y^{(0)}$.
Note that Theorem~\ref{th:bound} then provides a guarantee on the final solution, as it is at least as good as that of the maximum weight problem.
For line~\ref{algl:update}, one approach is to solve \eqref{eq:primal} with $\bm z$ fixed to $\bm z^{(t)}$ in order to get dual solution variables $\bm y^{(t)}$.
A better approach as observed in the experiments and explained next is based on Lagrangian relaxation.


The Lagrangian relaxation of \eqref{eq:primal} with respect to constraint  \eqref{eq:primalc} can be written as follows with Lagrangian multipliers $\bm \lambda = (\lambda_{ik})_{i \in [n], k \in [n]}$:
\begin{subequations}
\small
\label{eq:lagrangianrelaxation}
\begin{empheq}[]{align}
\mathcal L(\bm \lambda) &=& \!\!\!\!\!\!\!\!{\max}. ~ & \sum_{k \in [n]} (w'_k k - \sum_{i \in [n]} \lambda_{ik}) r_k - \sum_{k \in [n]} \sum_{i \in [n]} (w'_k - \lambda_{ik}) d_{ik}\label{eq:lagrangianrelaxation1}\\
&&& \quad\quad+ \sum_{i \in [n]} \big(\sum_{k \in [n]} \lambda_{ik} \big) \sum_{j \in [m]} u_{ij}z_{ij} \\
&&\mbox{s.t.~~} & \bm A \bm z \le \bm b  \\
&&&\bm z \in \{0, 1\}^{n \times m} \\
&&& d_{ik} \ge 0 \quad\quad \forall i \in [n], \forall k \in [n] \label{eq:lagrangianrelaxation2}
\end{empheq}
\end{subequations}
The Lagrangian dual of \eqref{eq:lagrangianrelaxation} is then given by:
\begin{align}
\label{eq:lagrangiandual}
{\min}. ~ \mathcal L(\bm \lambda) \quad
\mbox{s.t.~~} & \lambda_{ik} \ge 0 \quad \forall i \in [n], \forall k \in [n] 
\end{align}
For an optimal solution $\bm z^*, \bm r^*, \bm d^*$ of the $0,1$-LP \eqref{eq:primal}, we have for any $\bm \lambda \in \mathbb R^{n\times n}_+$:
\begin{align}
\sum_{k \in [n]} w'_k ( k r^*_k - \sum_{i \in [n]} d^*_{ik} ) &\le \sum_{k \in [n]} (w'_k k - \sum_{i \in [n]} \lambda_{ik}) r_k - \sum_{k \in [n]} \sum_{i \in [n]} (w'_k - \lambda_{ik}) d_{ik} \nonumber\\
&\qquad+ \sum_{i \in [n]} \big(\sum_{k \in [n]} \lambda_{ik} \big) \sum_{j \in [m]} u_{ij}z_{ij} \le \mathcal L(\bm \lambda) \nonumber
\end{align}
The first inequality holds because of the nonnegativity of $\bm \lambda$ and the feasibility of $\bm z^*, \bm r^*, \bm d^*$.
The second is true because of the maximization in \eqref{eq:lagrangianrelaxation}.
Therefore the best upperbound is provided by the solution of the Lagrangian dual \eqref{eq:lagrangiandual}, though this problem is not easy to solve due to the integrality condition over $\bm z$.

An inspection of program \eqref{eq:lagrangianrelaxation} leads to two observations: 
(i) it can be decomposed into two maximization problems, one over $\bm z$ and the other over $\bm r$ and $\bm d$;
(ii) for program \eqref{eq:lagrangianrelaxation} to yield a useful upperbound, $\bm \lambda$ should satisfy two constraints (otherwise $\mathcal L(\bm \lambda) = \infty$):
\begin{align*}
\sum_{i \in [n]} \lambda_{ik} = k w'_k \quad \forall k \in [n] \quad\mbox{and}\quad
\lambda_{ik} \le w'_k \quad \forall i \in [n], \forall k \in [n]
\end{align*}
Interestingly, in the above decomposition, the maximization problem over $\bm z$ corresponds to ($P_{\bm \lambda}$) and therefore $\bm \lambda$ can be identified to the dual variable $\bm y$.

Based on those observations, line~\ref{algl:update} can be performed as follows. 
Given $\bm \lambda$ (or $\bm y$), the upperbound $\mathcal L(\bm \lambda)$ can be improved by updating $\bm \lambda$ so as to decrease \eqref{eq:lagrangianrelaxation1}, which can be simply done by a projected sub-gradient step:
\begin{align}
\lambda'_{ik} &\gets \lambda_{ik} - \gamma (r_k - d_{ik} - \sum_{j \in [m]} u_{ij} z_{ij}) & \forall i\in [n], k\in [n]\\
\bm \lambda &\gets \operatorname*{arg\,min}_{\bm \lambda \in \mathbb L} || \bm \lambda' - \bm \lambda || \label{eq:projL}
\end{align}
where $\gamma$ is the sub-gradient step and \eqref{eq:projL} is the Euclidean projection of $\bm \lambda'$ on $\mathbb L = \{ \bm \lambda \in \mathbb R^{n \times n}_+ \,|\, \forall k \in [n], \sum_{i \in [n]} \lambda_{ik} = k w'_k, \forall i \in [n], \lambda_{ik} \le w'_k \}$.

Projection~\eqref{eq:projL} can be performed efficiently by exploiting the structure of $\mathbb L$:
\begin{align}
&\operatorname*{arg\,min}_{\bm \lambda \in \mathbb L} || \bm \lambda' - \bm \lambda || = \operatorname*{arg\,min}_{\bm \lambda \in \mathbb L} || \bm \lambda' - \bm \lambda ||^2 = \operatorname*{arg\,min}_{\bm \lambda \in \mathbb L} \sum_{i \in [n]} \sum_{k \in [n]} (\lambda'_{ik} - \lambda_{ik})^2 \nonumber\\
&= \big( \operatorname*{arg\,min}_{\bm \lambda_k \in \mathbb L_k} \sum_{i \in [n]} (\lambda'_{ik} - \lambda_{ik})^2 \big)_{k \in [n]} = \big( \operatorname*{arg\,min}_{\bm \lambda_k \in \mathbb L_k} \sum_{i \in [n]} (\frac{\lambda'_{ik}}{w'_k} - \frac{\lambda_{ik}}{w'_k})^2 \big)_{k \in [n]} \label{eq:projcap}
\end{align}
where $\mathbb L_k = \{ \bm \lambda_k \in \mathbb R^{n}_+ \,|\, \sum_{i \in [n]} \lambda_{ik}/w'_k = k, \forall i \in [n], \lambda_{ik}/w'_k \le 1 \}$.
Equation~\eqref{eq:projcap} states that projection~\eqref{eq:projL} can be efficiently performed by $n$ projections on capped simplices \cite{WangLu2015}.
The complexity of this step would be in $O(n^3)$, which is much faster than solving the quadratic problem \eqref{eq:projL}.
Besides, the $n$ projections can be easily computed in a parallel way.

We can provide a simple interpretation to the variable $\bm \lambda$ (or $\bm y$).
Considering programs~\eqref{eq:dual} and \eqref{eq:primalrelaxation}, we can observe that $\bm y$ corresponds to an allocation of weights $w'_k$'s over the different component $i$'s.
Indeed, an optimal solution of \eqref{eq:dual} would yield an extreme point of $\mathbb L$ (for a given $k \in [n]$, exactly $k$ terms among $(y_{1k}, \ldots, y_{nk})$ are equal to $w'_k$ and the other ones are null).
The projected sub-gradient method allows to search for an optimal solution of our fair combinatorial problem by moving inside the convex hull of those extreme points.

\vspace{-2mm}
\section{Experimental Results}\label{sec:expe}

\vspace{-2mm}
We evaluated our method on two different problems: assignment and matching. 
The LPs and  $0,1$-LPs were solved using CPLEX 12.7 on a PC (Intel Core i7-6700 3.40GHz) with 4 cores and 8 threads and 32 GB of RAM. 
Default parameters of CPLEX were used with 8 threads. The sub-gradient step $\gamma_t$ is computed following the scheme: $\gamma_t:=\frac{(val(\bm z_t)-bestvalue)\rho_t}{sqn}$ where $val(\bm z_t)$ is the objective value of the program 
\eqref{eq:primalrelaxation} with solution $\bm z_t$, $bestvalue$ is the best known objective value of the program \eqref{eq:primal} so far and $sqn$ is the square of the Euclidean norm of the subgradient vector. The parameter $\rho_t$ is divided by two every 3 consecutive iterations in which the upperbound $\mathcal L(\bm \lambda)$ has not been improved. 
The GGI weights were defined as follows: $w_k = 1/k^2$ for $k \in [n]$ so that they decrease fast in order to enforce more balanced solutions. \\[1ex]
{\bf Assignment}
To demonstrate the efficiency of our heuristic method, we generate hard random instances for the assignment problem. 
A random instance of this problem corresponds to a random generation of the $u_{ij}$'s, 
which are generated as follows.
For all $i \in [n]$, $u_{i1}$ follows a uniform distribution over $[100]$ and
for all $j \in [n]$, $u_{ij} = u_{i1} + \epsilon$ where $\epsilon$ is a random variable following a uniform distribution over integers between $-d$ and $d$ (with $d$ a positive integer parameter).
With such a generation scheme, agents' preferences over objects are positively correlated and the solution of the fair optimization problem is harder due to the difficulty of finding a feasible solution that satisfies everyone.\\[1ex]
\begin{table}[tb]
\centering
\begin{minipage}[t]{.47\textwidth}
\begin{tabular}[t]{|l|l|l|l|l|}
\hline
Instance & \multicolumn{2}{|c|}{CPLEX} & \multicolumn{2}{|c|}{AlterOpt}\\
\hline
 & CPU1 & CPU2 & CPU & Gap\\
\hline
v50-20 &1.02&1.02&0.23&0\%\\
v50-30 &3.14&3.14&0.26&0\%\\
v50-40 &64.95&14.26&0.45&0.28\%\\
v50-50 &1054.14&100.23&0.65&0.26\%\\
v30-20 &0.89&0.89&0.2&0\%\\
v30-30&8.83&8.83&0.3&0.015\%\\
v30-40&590.66&45.93&0.48&0.13\%\\
v10-20 & 1.55&1.55&0.18&0\%\\
v10-30 & 342.78&342.78&0.94&0\%\\
\hline
\end{tabular}%
\end{minipage}%
\begin{minipage}[t]{.47\textwidth}
\begin{tabular}[t]{|l|l|l|l|l|}
\hline
Instance & \multicolumn{2}{|c|}{CPLEX} & \multicolumn{2}{|c|}{AlterOpt}\\
\hline
 & CPU1 & CPU2 & CPU & Gap\\
\hline
v50-30 & 0.86 & 0.86 &0.79 & 0\%\\
\hline
v50-40 & 2.43 & 2.43 &1.42 & 0\%\\
\hline
v50-50 &5.14 &5.14 &2.67&0\%\\
\hline
v50-60 &148.5&25.45&13.43&0.01\%\\
\hline
v50-70 &2406.02& 1282.8&17.71&0.005\%\\
\hline
v30-30 & 1.15&1.15&0.78&0\%\\
\hline 
v30-40 &7.13 &7.13&1.44&0\%\\
\hline
v30-50&81.75&75.5&2.45&0.01\%\\
\hline
v30-60&1003.69&615.16&12.8&0.036\%\\
\hline
v10-30&5.33&5.33&0.76&0\%\\
\hline
v10-40&1325.7&806.8&1.4&0.06\%\\
\hline 
v10-50&29617.78&3370.7 &2.48&0.053\%\\
\hline
\end{tabular}
\end{minipage}%
\caption{Numerical results for (left) assignment and (right) general matching problems}
\end{table}
{\bf Matching}
We use the lemon library \cite{DezsJuttnerKovacs11} for solving the maximum weight matching problem.
For the generation of the matching problem (in a complete graph with $2n$ nodes), we follow a similar idea to the assignment problem.
Recall we only need $u_{ij}$ (and $z_{ij}$) for $i<j$.
For all $i \in [n]$, for all $j \in [n]$ with $i<j$, $u_{ij}=-1000$.
For all $i \in [n]$, $u_{i,n+1}$ follows a uniform distribution over $[100]$ and
for all $j \ge \max(i+1, n+2)$, $u_{ij} = u_{i,n+1} + \epsilon$ where $\epsilon$ is defined as above.\\[1ex]
%
{\bf Explanations}
The name of the instances is of the form "v$d$-$x$" where $d$ denotes the deviation parameter mentioned above and $x$ the number of the vertices of the graphs (i.e., $n=x/2$).  
Column ``CPLEX" regroups CPLEX's results. 
Subcolumn ``CPU1" reports the time (in seconds) that CPLEX spent to solve program \eqref{eq:primal} to optimal. 
Subcolumn ``CPU2" reports the times needed by the primal heuristic of CPLEX to obtain a feasible integer solution that is better than or equal to the solution given by our algorithm.
Column ``AlterOpt" reports our algorithm's results. 
Subcolumn ``CPU" is the time spent by our algorithm.
Subcolumn ``Gap" reports the gap in percentage between $Sol$ and $Opt$, which is equal to $(Opt -Sol)\times100/Opt \%$ where $Opt$ is the optimal value and $Sol$ is the value of the solution given by our algorithm.
The times and the gaps reported are averaged over $10$ executions corresponding to $10$ random instances.

Table 1 shows that the CPU time spent by CPLEX (subcolumn CPU1) for solving program \eqref{eq:primal} increases exponentially with $n$ and can quickly reach up to around 10 hours. 
Moreover, the smaller the deviation $x$, the more difficult the problem. 
For example, for $x=50$, we cannot solve instances with more than 50 and more than 70 vertices for respectively the fair assignment and general matching problems within 10 hours of CPU time. 
For $x=10$, this limit is respectively 30 and 50 vertices.
In contrast, the CPU time spent by our algorithm (subcolumn CPU) seems to increase linearly with $n$ and remains within tens or so seconds.
The quality of the solutions output by our algorithm is very good as the gap is at maximum around 0.3\% for fair assignment. 
This is even better for fair general matching, in all cases the gap is smaller than 0.1\%. 
Moreover, the CPU time that CPLEX needs to find a feasible integer solution of similar quality by primal heuristic is much longer than the CPU time of our algorithm (up to hundreds times longer).
It is interesting to notice that the fair assignment seems to be more difficult in our experiments than the fair general matching. 
This contrasts with the classical maximum weight version where the assignment problem is generally easier than the general maximum matching.



\section{Conclusion}
\label{sec:conclu}

We formulated the fair optimization with the Generalized Gini Index for a large class of combinatorial problem for which we proposed a primal-dual algorithm based on a Lagrangian decomposition.
We demonstrated its efficiency on several problems.
We also provided some theoretical bounds on its performance.
As future work, we plan to improve those bounds and investigate other updates for the Lagrangian multipliers.
Another interesting direction is to consider other linearization techniques such as the one proposed by Chassein and Goerigk \cite{ChasseinGoerigk15}.
Finally, we will also apply our method to problems whose maximum weight version is NP-hard.

\bibliography{biblio160226}
\bibliographystyle{splncs}
\end{document}